\documentclass{amsart}
\usepackage{graphicx}
\usepackage{bbm}

\newcommand{\eps}{\varepsilon}
\newcommand{\N}{\ensuremath{\mathbb{N}}}
\newcommand{\R}{\ensuremath{\mathbb{R}}}
\newcommand{\Z}{\ensuremath{\mathbb{Z}}}
\newcommand{\Q}{\ensuremath{\mathbb{Q}}}
\newcommand{\E}{\ensuremath{\mathbb{E}}}
\renewcommand{\P}{\ensuremath{\mathbb{P}}}
\newcommand\ind[1]{\ensuremath{\mathbbm{1}_{\{#1\}}}}
\newcommand{\diff}{\mathop{}\mathopen{}\mathrm{d}}
\newtheorem{proposition}{Proposition}
\newtheorem{lemma}{Lemma}
\newtheorem{theorem}{Theorem}
\newtheorem{corollary}{Corollary}

\newcommand{\Var}{\mathop{}\mathopen{}\mathrm{Var}}

\newcommand\croc[1]{\left\langle #1\right\rangle}
\newcommand\cal{\mathcal}
\title{Asymptotics of Stochastic Protein Assembly Models} 

\author{Marie Doumic}
\email{Marie.Doumic@inria.fr}
\urladdr{https://team.inria.fr/mamba/marie-doumic/}

\author{Sarah Eug\`ene}
\email{Sarah.Eugene@inria.fr}

\author{Philippe Robert}
\email{Philippe.Robert@inria.fr}
\urladdr{http://team.inria.fr/rap/robert}

\address[M. Doumic, S. Eug\`ene, Ph. Robert]{INRIA Paris, 2 rue Simone Iff, F-75012 Paris, France}
\address[M. Doumic, S. Eug\`ene]{Sorbonne Universit\'es, UPMC Universit\'e Pierre et Marie Curie, UMR 7598, Laboratoire Jacques-Louis Lions, F-75005, Paris, France}

\date{\today}

\begin{document}

\begin{abstract}
Self-assembly of proteins  is a biological phenomenon which gives rise to  spontaneous formation of  {\em amyloid fibrils}  or {\em polymers}.   The  starting point  of this  phase, called \emph{nucleation} exhibits an important variability among replicated experiments.  To analyse the stochastic nature of this phenomenon, one of the simplest models  considers two populations of chemical components:  monomers and polymerised monomers. Initially there are only  monomers.  There are two reactions for the polymerization of a monomer: either two monomers collide to combine into two polymerised monomers or a monomer is  polymerised after the encounter of a polymerised monomer. It turns out that this simple model does not  explain completely the variability observed in the experiments. This paper investigates  extensions of this model  to  take into account other mechanisms of the polymerization process that may have impact an impact on fluctuations.  The first variant consists in introducing a preliminary conformation step to take into account the biological fact that, before being polymerised, a  monomer has two states,  regular or misfolded.  Only misfolded monomers can be polymerised so that the fluctuations of the number of misfolded monomers can be also a source of variability of the number of polymerised monomers.  The second variant, based on numerical considerations,  represents the reaction rate $\alpha$   of spontaneous formation of a polymer as of  the order of $N^{-\nu}$, for some large scaling variable $N$ representing the reaction volume and $\nu$ some positive constant.  Asymptotic results involving different time scales are obtained for the corresponding Markov processes. First and second order results for the starting instant of nucleation are derived from these limit theorems.  The proofs of the results rely on a study of a stochastic averaging principle for a model related to an Ehrenfest urn model, and also on a scaling analysis of a population model. 
\end{abstract}

\maketitle 

\hrule

\vspace{-3mm}

\tableofcontents

\vspace{-10mm}

\hrule

\bigskip

\section{Introduction}
Self-assembly of proteins  is an important biological phenomenon, on the one hand associated with human diseases such as Alzheimer's, Parkinson's, Huntington's diseases and still many others, and on the other hand involved in industrial processes, see McManus et al.~\cite{McManus_2016}  and Ow and Dustan~\cite{Ow_Protein2014}. The initial step of the chain reactions giving rise to amyloid fibrils  consists in the spontaneous formation of a so-called \emph{nucleus}, that is, the simplest possible polymer able to ignite the reaction. This early phase is called \emph{nucleation}, and is still far from being understood. As underlined by previous studies Szavits-Nossan et al.~\cite{Szavits}, the nucleation step is intrinsically stochastic, leading to an important variability among replicated experiments, not only in small volumes but even in relatively large ones, see Xue et al.~\cite{Radford}. The question of building convenient stochastic models, able to render out the heterogeneity  observed, and even to predict it, has recently raised much interest in the biological and biophysical community, see Szavits-Nossan et al.~\cite{Szavits}, Yvinec et al.~\cite{yvinec2016}, Pigolotti et al.~\cite{Pigolotti2013} and Eden et al.~\cite{Eden2015}.

We start with a simple  stochastic model,  proposed and studied in Eug\`ene et al.~\cite{EXRD} for which we consider extensions to get a deeper understanding on the intricate influence of each reaction considered. In Eug\`ene et al.~\cite{EXRD}, rigorous asymptotics of the simple model were proved, and it was fitted to the experimental data published in Xue et al~\cite{Radford}.  It was shown that  the predicted variability was much smaller by the model than what was experimentally obtained. One of the conclusions of this work is that other mechanisms had to be taken into account to explain the variability observed in the experiments.
We thus propose here two ways to complement the basic model. Let us first recall its definition.
\subsection{The Basic Model}
One of the simplest models to describe  the nucleation process considers two populations of chemical components: (regular) monomers and polymerised monomers. Initially there are only  monomers.  There are two reactions for the polymerization of a monomer: either two monomers collide to combine into two polymerised monomers or a monomer is  polymerised after the encounter of a polymerised monomer. The chemical reactions associated with the basic model can then be described as follows:
\begin{equation}\label{ChemBas}
\begin{cases}
{\cal X}_1+{\cal X}_1 \stackrel{{\alpha}}{\underset{}{\longrightarrow}} 2 {\cal X}_2,\\
{\cal X}_1+{\cal X}_2 \stackrel{{\beta}}{\underset{}{\longrightarrow}} 2 {\cal X}_2.
\end{cases}
\end{equation}
These reactions can be represented by the sample paths of a Markov process $(X_1^N(t),X_2^N(t))$, where $X_1^N(t)$ [resp. $X_2^N(t)$] is the number of regular [resp. polymerised]  monomers at time $t\geq 0$. 
The scaling variable $N$ which will be used should be thought of as the reaction volume. In particular $X_2^N(t)/N$ is the concentration of polymerised monomers at time $t$.  If $M_N$ is the initial number of monomers, it is assumed that the following regime
\begin{equation}\label{scaling}
\lim_{N\to+\infty}\frac{M_N}{N}=m
\end{equation}
holds for some $m>0$. The quantity $M_N/N$ is in fact the initial concentration of monomers. 

The transition rates of $(X_1^N(t),X_2^N(t))$ are given by, for $x=(x_1,x_2)\in\N^2$,
\begin{equation}\label{RateTM}
 x \mapsto 
\begin{cases}
x{+}({-}2,2) &\text{ at rate } \quad \alpha (x_1/N)^2\\
x{+}({-}1,1)  &\phantom{ at rate } \quad \beta x_1/N\times x_2/N.
\end{cases}
\end{equation}
The second coordinate $x_2$ is in fact the polymerized mass which explains the jumps of size $2$ in the reactions.
Note that the conservation of mass implies that the quantity $X_1^N(t){+}X_2^N(t)$ is constant and equal to $M_N$, the total number of initial monomers. 
\begin{enumerate}
\item[---] The first reaction of~\eqref{ChemBas} converts two monomers into two polymerised monomers. In our model, due to thermal noise in particular, these reactions will occur in a stochastic way.  Following the principles of the law of mass action,  the encounter of two chemical species  occurs  at a rate  proportional to the product of the \emph{concentrations} of each species.  Therefore two given monomers disappear to produce two polymerised monomers  at a rate $\alpha (x_1/N)^2$. 
\item[---] The second reaction can be seen as an auto-catalytic process. Here, given a monomer at the contact of a polymerised monomer, the monomer  is converted into a polymerised monomer at a rate ${\beta}$.  Again, by the law of mass action, regular monomers disappear at the rate $\beta(x_1/N)(x_2/N)$.
\end{enumerate}
See Eug\`ene et al.~\cite{EXRD}, Szavits-Nossan et al.~\cite{Szavits} and  Xue et al.~\cite{Radford}  for a general presentation of these phenomena in a biological context.  For more  discussion and  results on stochastic models associated to chemical reactions, see  for example  Anderson and Kurtz~\cite{AndersonKurtz} and Higham~\cite{Higham} and references therein.

This simple, intuitive model of polymerisation has the advantage of having only two parameters to determine.  It can be analyzed mathematically by standard tools of probability theory, see Eug\`ene et al.~\cite{EXRD}. It has been shown that if $X_2^N(t)$ is the number of polymerised monomers at time $t$, then the polymerisation process can be described  via the following convergence in distribution 
\begin{equation}\label{TMCV}
\lim_{N\to+\infty}\left(\frac{X_2^N(Nt)}{N}\right)=(x_2(t))
\end{equation}
holds, where $(x_2(t))$ is the non-trivial solution of the following simple ordinary differential equation 
\begin{equation}\label{eq:simple}
\dot{x}_2 (t) =\alpha \big(m-x_2(t)\big)^2 + \beta \big(m-x_2(t)\big)x_2(t)
\end{equation}
 converging to $m$ has $t$ goes to infinity. 

By using  these simple mathematical results and the data from experiments with 17 different concentrations of monomers (the value of $m$)  and 12  experiments for each concentration, Table~I of Eug\`ene et al.~\cite{EXRD} shows that, in this setting, the estimation of $\beta$ is reasonably robust. This is unfortunately not  the case for the numerical estimation of $\alpha$ which is varying from $1.68{\cdot}10^{-2}$ to $9.57{\cdot}10^{-8}$. An additional difficulty with this simple model comes from the small values of $\alpha$ obtained. Indeed, for the experiments, the value of the volume $N$ is in the order of $10^{15}$,  some of the estimated values of $\alpha$ in $10^{-8}$ are therefore, numerically,  of the order of $1/\sqrt{N}$. The asymptotic results are obtained when $N$ gets large and $\alpha$ {\em fixed}.  For this reason, one may suspect a problem of convergence speed in Relation~\eqref{TMCV} when these parameters are used. It turns out  that our simulations confirm that the asymptotic regime~\eqref{TMCV} does not seem to represent accurately the system when $\alpha$ is too small.

The purpose of the present paper is to refine this basic model in two different ways. 
\begin{enumerate}
\item The model can be improved by introducing a key feature of the polymerisation process: the misfolding of monomers.
Experiments show that monomers can be polymerised only if their $3$-D structure has been modified by some events. Such monomers are called misfolded monomers, see Dobson~\cite{Dobson,Dobson:2}, Knowles et al.~\cite{Knowles:2}. It turns out that, at a given time, only a small fraction of monomers are misfolded which may also explain that the polymerisation process starts very slowly. In biological cells, this phenomenon of misfolding is reversible, dedicated proteins may ``correct'' the misfolded monomers. A misfolded monomer can be turned into a ``regular'' monomer and vice-versa. See Bozaykut et al.~\cite{Boza} and Lanneau et al.~\cite{Lanneau} for example. Section~\ref{AlpSec} is devoted to the mathematical analysis of these models. 
\item Another approach is to keep the basic model but with the parameter $\alpha$ being of the order of $1/N^\nu$ for some positive $\nu$ to take into account that, in practice, the values of this parameter can be very small. Note that this is only a numerical observation, the value of $\alpha$ has no reason to depend on the volume. This model is analyzed in Section~\ref{AlpSec}. 
\end{enumerate}
The rest of the section is devoted to a brief sketch of the mathematical aspects of these two classes of models. As it will be seen, the models are more challenging from a mathematical point of view, the model with misfolded monomers in particular. 

\subsection{Models with Misfolding Phenomena}
The chemical reactions associated with this simple model are as follows:
\[
{\cal X}_0 \stackrel{\gamma}{\underset{\displaystyle\stackrel{\longleftarrow}{\gamma^*}}{\longrightarrow}}  {\cal X}_1,\qquad 
\begin{cases}
{\cal X}_1+{\cal X}_1 \stackrel{{\alpha}}{\underset{}{\longrightarrow}} 2 {\cal X}_2,\\
{\cal X}_1+{\cal X}_2 \stackrel{{\beta}}{\underset{}{\longrightarrow}} 2 {\cal X}_2. 
\end{cases}
\]
At time $t\geq 0$,  $X^N_0(t)$ denotes the number of regular  monomers,  $X^N_1(t)$ the number of misfolded monomers. As before the last coordinate $X^N_2(t)$ is the polymerized mass.  As a Markov process, $(X^N(t))=(X^N_0(t),X^N_1(t),X^N_2(t))$ has the following transitions, for an element $x=(x_0,x_1,x_2)\in\N^3$,
\begin{equation}\label{RateMF}
x \mapsto 
\begin{cases}
x{+}(1,{-}1,0) \text{ at rate } \gamma^*\,x_1\\
x{+}({-}1,1,0)  \phantom{ atsa rate } \gamma\,x_0,
\end{cases}
 x \mapsto 
\begin{cases}
x{+}(0,{-}2,2) & \alpha\,(x_1/N)^2\\
x{+}(0,{-}1,1)  & \beta\,x_1/N\times x_2/N.
\end{cases}
\end{equation}
It is important to note that the transition  between state ``0'', regular monomer, and state ``1'', misfolded monomer, is spontaneous. Consequently, as it can be seen, the corresponding transition rates {\em do not} depend on the volume $N$  but simply on the numbers of components and not on their concentrations. An important consequence of this observation is that the system  exhibits a two time scales behavior that we will investigate. 
\subsection*{An informal description of the asymptotic behavior of $(X_2^N(t))$}

The first two coordinates can be seen as an Ehrenfest process with two urns $0$ and $1$ where each particle in urn $0$ (resp. $1$) goes to urn $1$ (resp. $0$) at rate $\gamma$ (resp. $\gamma^*$).  See Bingham~\cite{Bingham} and Karlin and McGregor~\cite{Karlin} for example. Particles in urn $1$ can also go to the urn $2$ corresponding to the polymerized mass but this phenomenon occurs at a much slower rate so that, locally, it does not change the orders of magnitude in $N$ of $X_2^N$.

 When  $X_2^N{\sim}x_2N$, there is a total of $(m{-}x_2)N$ particles in the urns $0$ or $1$.  The  components $(X_0^N(t),X_1^N(t))$  are both of the order of $N$ and are moving on a fast time scale, proportional to $N$.  The transition rates of the process $(X_2^N(t))$ are slower, bounded with respect to $N$. Because of the fast transition  rates of the first two coordinates, the Ehrenfest urn process should reach quickly an equilibrium for which $X_0^N$ has a binomial distribution with parameter $(m-x_2)N$ and $r$ with $r={\gamma}/{(\gamma{+}\gamma^*)}$, in particular
\[
\frac{X_0^N}{N}\sim (1-r) (m-x_2)\text{ and } \frac{X_1^N}{N}\sim r(m-x_2).
\]
This suggests that, 
\begin{enumerate}
\item[a)] to see an evolution of $X_2^N$ of the order of $N$, one has to be on the linear time scale $t\mapsto N t$: transition rates of the process $X_2^N$ are $O(1)$), 
\item[b)] if $X_2^N(Nt)\sim x_2(t)N$, in view of transition rates of $(X_2^N(t))$ of Relation~\eqref{RateMF}, then  $(x_2(t))$ should satisfy the following ordinary differential equation
\begin{equation}\label{eq:intermediaire}
\dot{x}_2(t)=\alpha r^2(m-x_2(t))^2+\beta r(m-x_2(t))x_2(t). 
\end{equation}
\end{enumerate}
We recognize the limit equation~\eqref{eq:simple} of the simple model, where $\alpha$, $\beta$ are respectively replaced by $\alpha r^2$ and $\beta r$. This result is also true when considering the second order fluctuations of the number of polymers, see Theorem~\ref{CLTMF}. The proof of the convergence  of the process of the concentration of polymerized monomers to the solution of the ODE~\eqref{eq:simple} use standard arguments of convergence of a sequence of stochastic processes, see the supplementary material of Eug\`ene et al.~\cite{EXRD}.  The proof of the corresponding result  with misfolding phenomena for the ODE~\eqref{eq:intermediaire} is, as we shall see, more delicate to handle. 
\subsection*{Stochastic Averaging Phenomenon} To summarize these observations, the coordinates $(X_0^N(t),X_1^N(t))$ form a ``fast'' process and $(X_2^N(t))$ is a ``slow'' process when the scaling parameter $N$ goes to infinity.  This suggests a stochastic averaging principle (SAP) in a fully coupled context.
\begin{enumerate}
\item The stochastic evolution of $(X_2^N(Nt))$ is driven by the invariant distribution of an ``instantaneous'' associated Ehrenfest process.
\item The parameters of the Ehrenfest process depend on the macroscopic variable $(X_2^N(Nt))$.
\end{enumerate}
see Papanicolaou et al.~\cite{PSV} and Chapter~8 of Freidlin and Wentzell~\cite{Freidlin} for example, see also Kurtz~\cite{Kurtz}. 

A stochastic averaging principle is indeed proved as well as a corresponding central limit theorem (CLT). In our cases there are some differences with the ``classical'' framework of stochastic averaging principles. The state space of the fast process depends on the scaling parameter $N$, and is not in particular a ``fixed'' process  (with varying parameters) as it is usually the case. See Hunt and Kurtz~\cite{Hunt} or Sun et al.~\cite{Sun} for example.  A law of large numbers with respect to $N$ for the invariant distribution of the fast process is driving the evolution of the slow process. The approach used in the paper relies on the use of occupation measures on a continuous state space instead of a discrete space, this leads to some technical complications as it will be seen.  Concerning central limit theorems in a SAP context, there are few references available for jump processes. The methods presented in Kang et al.~\cite{KKP} or in Sun et al.~\cite{Sun} do not seem to be helpful in our case. Instead, an ad-hoc estimation, Proposition~\ref{propclt}, gives the main ingredient to derive a central limit theorem, 
see Section~\ref{AVGsec}. 

\subsection{Models with Scaled Reaction Rates}
Again, $X_1^N(t)$ (resp. $X_2^N(t)$) is the number of regular (resp. polymerised)  monomers at time $t\geq 0$.
The transition rates of the Markov process  $(X^N(t)){=}(X_1^N(t),X_2^N(t))$ associated to these models are the same, except that the parameter $\alpha$ is replaced by $\alpha/N^\nu$  with $0<\nu.$ For  $x=(x_1,x_2)\in\N^2$, the rates are given by 
\begin{equation}\label{RateSM}
 x \mapsto 
\begin{cases}
x{+}({-}2,2) &\text{ at rate } \quad  \alpha/N^\nu\,(x_1/N)^2\\
x{+}({-}1,1)  &\phantom{ at rate } \quad  \beta\,x_1/N\times x_2/N.
\end{cases}
\end{equation}
Convergence~\eqref{TMCV} shows that the polymerisation occurs on the linear time scale $t\mapsto Nt$ for the basic model. It will be shown that the phenomenon does not start on this time scale. A slightly more rapid time scale is necessary for this purpose, it is shown that it is on the time scale $t\mapsto N\log N{\cdot}t$ for $0<\nu\leq 1$ and $t\mapsto N^\nu t$ when $\nu{>}1$. See Section~\ref{AlpSec}.

\section{Stochastic Models with Misfolding Phenomena}\label{AVGsec}
The following notations will be used throughout the paper. For $\xi\geq 0$,   $\mathcal{N}_{\xi}(\diff t)$ denotes a Poisson process with  parameter $\xi$ and $(\mathcal{N}_{\xi}^i(\diff t))$ an i.i.d. sequence of such processes. All the Poisson processes are defined on a  probability space $(\Omega,{\cal F},\P)$.   If $f$ is a real valued function on $\R_+$, $f(t{-})$ denotes its limit on the left of $t{\geq}0$ when it exists.  Finally, $m^*$ denotes an upper bound for the sequence $(M_N{/}N)$ which converges to $m>0$ by Relation~\eqref{scaling}.

Recall that, at time $t\geq 0$ ,  $(X^N(t)){=}(X^N_0(t),X^N_1(t),X^N_2(t))$ where $X^N_0(t)$ is the number of  monomers,  $X^N_1(t)$ is the number of misfolded monomers and $X^N_2(t)$ is the polymerized mass. 
It is not difficult to see that these processes can be seen as the solution of the following stochastic differential equations, 
\begin{equation}\label{SDEMF}
\begin{cases}
\displaystyle\diff X_0^N(t) =  \sum_{i=1}^{X_1^N(t{-})} \mathcal{N}_{\gamma^*}^i(\diff t){-}\sum_{i=1}^{X_0^N(t{-})} \mathcal{N}_{\gamma}^i(\diff t), \\
\displaystyle\diff X_2^N(t) =  2 \sum_{i=1}^{{X_1^N(t-)(X_1^N{-}1)(t{-})}/{2}} \mathcal{N}_{{\alpha}/{N^2}}^i(\diff t){+}\sum_{i=1}^{X_1^N(t{-})X_2^N(t-)} \mathcal{N}_{{\beta}/{N^2}}^i(\diff t), 
\end{cases}
\end{equation}
with the relation of conservation of mass $M_N{=}X_0^N(t){+}X_1^N(t){+}X_2^N(t)$ and initial condition $X^N(0){=}(M_N,0,0)$. 

Equation~\eqref{SDEMF} gives in particular that
\begin{multline}\label{SDEX2}
X_2^N(t)=X_2^N(0)+\frac{\alpha}{N^2}\int_0^t X_1^N(s)(X_1^N(s){-}1)\,\diff s\\+\frac{\beta}{N^2}\int_0^t X_1^N(s) X_2^N(s)\,\diff s+M_2^N(t),
\end{multline}
where $(M_2^N(t))$ is a martingale whose previsible increasing process is given by
\begin{equation}\label{X2croc}
\croc{M_2^N}(t)=2\frac{\alpha}{N^2}\int_0^t X_1^N(s)(X_1^N(s){-}1)\,\diff s+\frac{\beta}{N^2}\int_0^t X_1^N(s) X_2^N(s)\,\diff s
\end{equation}

For $i=0$, $1$, $2$ and $t\geq 0$, denote
\[
\overline{X}_i^N(t)=\frac{X_i^N(Nt)}{N},
\]
the main goal of this section is to prove that the process $(\overline{X}_2^N(t))$ is converging in distribution to the solution $(x_2(t))$ of a non-trivial ordinary differential equation. It will show in particular that the polymerization process is occurring on the linear time scale $t\mapsto N t$.

\subsection{Random Measures Associated to Occupation Times}
Define $\mu_N$ the random measure on $\R_+^3$ by
\[
\croc{\mu_N,g}=\int_{\R_+} g\left(\overline{X}_0^N(Nu),\overline{X}_1^N(Nu),u\right)\,\diff u.
\]
\begin{proposition}\label{propmu}
The sequence $(\mu_N)$ is tight. Any limiting point $\mu_\infty$ of this sequence is such that
\begin{equation}\label{eqrep}
\croc{\mu_\infty,g}=\int_{\R_+^3} g\left(x,y,u\right)\,\pi_u(\diff x,\diff y)\diff u,
\end{equation}
for any continuous function $g$ on $[0,m^*]^2\times [0,T]$, where for each $u\geq 0$, $\pi_u$ is a random Radon measure on $\R_+^2$. 
\end{proposition}
\begin{proof}
Since $\overline{X}_0^N(t)$ and $\overline{X}_1^N(t)$ are bounded, for any $T>0$, the measure $\mu_N$ restricted to the set $\R_+^2\times[0,T]$ has a compact support.  Lemma~3.2.8~page~44 of~Dawson~\cite{Dawson} gives directly that the sequence $(\mu_N)$ of random measure on $\R_+^3$ is tight. 

 Let $(\mu_{N_k})$ be a convergent subsequence with limit $\mu_\infty$. By using Skorohod's representation theorem, one can assume that there exists a negligible measurable set ${\cal A}$ of the probability space such that, outside this subset, the convergence of the sequence $(\mu_{N_k})$ of Radon measures towards $\mu_\infty$, that is 
\[
\lim_{k\to+\infty} \croc{\mu_{N_k},g} =\croc{\mu_\infty,g} \text{ for all } g\in C([0,m^*]^2\times[0,T]), 
\]
holds. 

Let $h\in C([0,m^*]^2)$ and $f\in C([0,T]$, denoting $h\otimes f (x,y,u)=h(x,y)f(u)$, for $(x,y)\in[0,m^*]^2$ and $u\in[0,T]$, then, as a limit of the sequence $(\mu_{N_k})$,  the Radon measure
\[
f\mapsto \croc{\mu_N,h\otimes f}
\] 
is absolutely continuous with respect to Lebesgue's measure. Consequently, for any $h\in C([0,m^*]^2)$, there exists some function $(\tilde{\pi}_u(h),0\leq u\leq T)$ such that
\[
\croc{\mu_\infty,h\otimes f}=\int_0^T\tilde{\pi}_u(h)f(u)\,\diff u.
\]
By the differentiation theorem, see Theorem~7.10 in Rudin~\cite{Rudin}, the function $(\tilde{\pi}_u(h))$ can be represented as 
\[
\tilde{\pi}_u(h)=\limsup_{\eps\to 0}\frac{1}{\eps} \croc{\mu_\infty,h\otimes \ind{[u-\eps/2,u+\eps/2]}}, \quad u\in[0,T],
\]
consequently, the mapping $(\omega,u)\mapsto \tilde{\pi}_u(h)(\omega)$ is ${\cal F}\otimes {\cal B}([0,T])$-measurable. 

Let ${\cal S}$ be a countable dense subset of $C([0,m^*]^2)$, then there exists a subset $E_0$ of $[0,T]$ negligible for the Lebesgue measure such that, for all $u\in[0,T]\setminus E_0$ and $\phi_1$, $\phi_2\in {\cal S}$, 
\begin{enumerate}
\item $\tilde{\pi}_u(p_1\phi_1 + p_2\phi_2)=p_1\tilde{\pi}_u(\phi_1) + p_2\tilde{\pi}_u(\phi_2)$, $\forall p_1$, $p_2\in\Q$,
\item $\tilde{\pi}_u(\phi_1)\leq \tilde{\pi}_u(\phi_2)$ if $\phi_1\leq \phi_2$,
\item $\tilde{\pi}_u(1)=1$.
\end{enumerate}
With the same method as in Section~II.88 of Rogers and Williams~\cite{Rogers},  for any $u\in[0,T]\setminus E_0$ , one gets the existence of a Radon measure $\pi_u$ on $[0,m^*]^2$ such that
$\tilde{\pi}_u(h)=\pi_u(h)$ for any $h\in {\cal S}$. By density of ${\cal S}$, the mapping $(\omega,u)\mapsto {\pi}_u(h)(\omega)$ is  also ${\cal F}\otimes {\cal B}([0,T])$-measurable and the relation
\[
\croc{\mu_\infty,h\otimes f}=\int_0^T{\pi}_u(h)f(u)\,\diff u.
\]
holds for all $h\in C([0,m^*]^2)$ and $f\in C([0,T]$. The proposition is therefore proved. 
\end{proof}
Representation~\eqref{eqrep} is related to Lemma~1.4 of Kurtz~\cite{Kurtz}.   Our proof relies on classical arguments of measure theory, a functional version of Carath\'eodory's extension theorem in particular which is  described in Section~II.88 of Rogers and Williams~\cite{Rogers}. In Kurtz~\cite{Kurtz},  a more sophisticated result,  see Morando~\cite{Morando}, on the extension of bi-measures is the key ingredient. The notion of bi-measure goes back to Kingman, see Dellacherie and Meyer~\cite{DM} for example. It should be mentioned that  Lemma~1.4 of Kurtz~\cite{Kurtz} gives also additional measurability properties of the family $(\pi_u)$ which are of no use in our case.
\begin{proposition}\label{proppi}
If  $\mu_\infty$ is a limiting point of $(\mu_N)$ with the representation~\eqref{eqrep}  then, for any  ${\cal C}^1$-function $f$ on $\R_+^2$, almost surely
\begin{equation}\label{idav}
\int_0^t \int_{\R_+^2} \left(\gamma^*y-\gamma x\right)\left(\frac{\partial}{\partial x}f(x,y)-\frac{\partial}{\partial y}f(x,y)\right) \pi_u(\diff x,\diff y)\,\diff u=0, \quad \forall t\geq 0,
\end{equation}
in particular, almost surely,
\begin{equation}\label{idSAP}
\int_0^t \int_{\R_+^2}\left(\gamma^*y-\gamma x\right)^2 \pi_u(\diff x,\diff y)\,\diff u=0, \quad \forall t\geq 0. 
\end{equation}
\end{proposition}
Relation~\eqref{idSAP} just says that almost surely and for almost all $u$, the measure $\pi_u$ is degenerated on $\R_+^2$ and carried by the subset$\{(x,\gamma x/\gamma^*):0\leq x\leq m\}$. 
\begin{proof}
for $(i,j)\in\Z^2$, one denotes  by $\Delta_{ij}$ the discrete differential operator
\[
\Delta^N_{ij}(f)(x,y)= f(x+i/N,y+j/N)-f(x,y), \quad (x,y)\in[0,m^*]^2. 
\]
After some trite calculations, the stochastic differential equations~\eqref{SDEMF} give the relation
\begin{multline}\label{eqaux1}
f\left(\overline{X}^N(t/N)\right) =f\left(\overline{X}^N(0)\right)+  \gamma\int_0^t X_0^N(s) \Delta^N_{-1,1}(f)\left(\overline{X}^N(s/N)\right) \,\diff s\\
+ \gamma^*\int_0^t X^N_1(s) \Delta^N_{1,-1}(f)\left(\overline{X}^N(s/N)\right)\, \diff s\\
+ \alpha \int_0^t \frac{X^N_1(s)(X^N_1(s)-1)}{2N^2} \Delta^N_{0,-2}(f)\left(\overline{X}^N(s/N)\right)\, \diff s\\
+\beta \int_0^t \frac{X^N_1(s)}{N}\frac{X^N_2(s)}{N}  \Delta^N_{0,-1}(f)\left(\overline{X}^N(s/N)\right)\, \diff s+M_f^N(t),
\end{multline}
where $(\overline{X}^N(t))=({{X}_0^N(Nt)/N},{{X}_1^N(Nt)/N})$ and $(M_f^N(t))$ is the  associated martingale. Its previsible increasing process is given by
\begin{multline}\label{eqaux2}
\croc{M_f^N}(t)=\gamma\int_0^t X_0^N(s) \Delta^N_{-1,1}(f)^2\left(\overline{X}^N(s/N)\right) \,\diff s\\
+ \gamma^*\int_0^t X_1(s) \Delta^N_{1,-1}(f)^2\left(\overline{X}^N(s/N)\right)\, \diff s\\
+ \alpha \int_0^t \frac{X_1^N(s)(X_1^N(s)-1)}{2N^2}\Delta^N_{0,-2}(f)^2\left(\overline{X}^N(s/N)\right)\, \diff s\\
+\beta \int_0^t \frac{X^N_1(s)}{N}\frac{X^N_2(s)}{N}\Delta^N_{0,-1}(f)^2\left(\overline{X}^N(s/N)\right)\, \diff s.
\end{multline}
Note that, for $i$, $j\in\Z$
\[
\Delta_{i,j}^N(f)(x,y)=\frac{1}{N}\left(i\frac{\partial f}{\partial x}(x,y)+j\frac{\partial f}{\partial y}(x,y)\right)+o(1/N),
\]
by changing the time variable in $Nt$ in Equation~\eqref{eqaux1} and by dividing by $N$ one gets the relation
\begin{multline}\label{eqaux3}
\frac{1}{N}\left(f\left( \overline{X}^N(t)\right) -f\left( \overline{X}^N(0)\right)\right)\\
= \int_0^t\left[  \gamma^*\overline{X}_1^N(s)-\gamma \overline{X}_0^N(s) \right]  \left[\frac{\partial f}{\partial x}-\frac{\partial f}{\partial y}\right]\left(\overline{X}^N(s)\right)\,\diff s\\
- \frac{\alpha}{N} \int_0^t \overline{X}^N_1(s)\left(\overline{X}^N_1(s)-1/N\right) \frac{\partial f}{\partial y}\left(\overline{X}^N(s)\right)\, \diff s\\
-\frac{\beta}{N} \int_0^t \overline{X}^N_1(s)\overline{X}^N_2(s) \frac{\partial f}{\partial y}\left(\overline{X}^N(s)\right)\, \diff s+\frac{M_f^N(Nt)}{N}+o(1/N),
\end{multline}
with $(\overline{X}^N(t)=(\overline{X}^N_0(t),\overline{X}^N_1(t))$.
The previsible increasing process of the martingale $(M_f^N(Nt)/N)$ in the above expression is $(\langle M_f^N\rangle(Nt)/N^2)$. By using Equation~\eqref{eqaux2} and the fact that $(\overline{X}_i^N(t))$ is bounded for $i=0$ and $1$, it is not difficult to show that its expected value converges to $0$ as $N$ gets large and, by Doob's Inequality, that the martingale  converges in distribution to $0$.  With similar arguments, from Equation~\eqref{eqaux3}, one gets therefore the following convergence in distribution
\begin{equation}\label{eqaux4}
\lim_{N\to+\infty} \left(\int_0^t\left[  \gamma^*\overline{X}_1^N(s)-\gamma \overline{X}_0^N(s) \right]  \left[\frac{\partial f}{\partial x}-\frac{\partial f}{\partial y}\right]\left(\overline{X}^N(s)\right)\,\diff s\right)=0.
\end{equation}
For $t\geq 0$, 
\begin{align*}
\int_0^t\left[  \gamma^*\overline{X}_1^N(s)-\gamma \overline{X}_0^N(s) \right]&\left[\frac{\partial f}{\partial x}-\frac{\partial f}{\partial y}\right]\left(\overline{X}^N(s)\right)\,\diff s\\
&=\int \left[  \gamma^*y-\gamma x \right]  \left[\frac{\partial f}{\partial x}-\frac{\partial f}{\partial y}\right]\left(x,y\right)\ind{s\leq t}\,\mu_N(\diff x,\diff y, \diff s),
\end{align*}
and this last term converges in distribution to
\begin{multline*}
\int \left[  \gamma^*y-\gamma x \right]  \left[\frac{\partial f}{\partial x}-\frac{\partial f}{\partial y}\right]\left(x,y\right)\ind{u\leq t}\,\mu_\infty(\diff x,\diff y, \diff u)
\\=\int_0^t \int_{\R_+^2} \left(\gamma^*y-\gamma x\right)\left(\frac{\partial}{\partial x}f(x,y)-\frac{\partial}{\partial y}f(x,y)\right) \pi_s(\diff x,\diff y)\,\diff s.
\end{multline*}
This convergence in distribution also holds for any finite marginals of this process. The convergence of processes~\eqref{eqaux4} gives therefore the desired identity~\eqref{idav} in distribution. 
The last assertion of the proposition is proved by taking the function $f(x,y)=\gamma^*y^2-\gamma x^2$. 
\end{proof}
\subsection{A Stochastic Averaging Principle}
Relation~\eqref{SDEX2} gives the following integral equation for $(\overline{X}_2^N(t))$, 
\begin{multline}\label{SDEX22}
\overline{X}_2^N(t))=\overline{X}_2^N(0)+\alpha\int_0^t \overline{X}^N_1(s)\left(\overline{X}_1^N(s){-}1/N\right)\,\diff s\\+\beta\int_0^t \overline{X}_1^N(s) \overline{X}_2^N(s)\,\diff s+\frac{M_2^N(Nt)}{N},
\end{multline}
The expected value  of the previsible increasing process of the martingale converges ${M_2^N(Nt)}{N})$ is vanishing  as $N$ gets large by Equation~\eqref{X2croc}.   Doob's Inequality shows that the martingale  converges in distribution to $0$. The criteria of the modulus of continuity, see Billingsley~\cite{Billingsley}, gives therefore that  the sequence of processes $(\overline{X}_2^N(t))$ is tight.  It can therefore be assumed, for some  subsequence $(N_k)$, that the following convergence holds,
\[
\lim_{k\to+\infty}\left(\mu_{N_k},\left(\overline{X}_2^{N_k}(t)\right)\right)=(\mu_\infty,(x_2(t)))
\]
for a random measure $\mu_\infty$ as in Proposition~\ref{propmu} and some continuous stochastic process $(x_2(t))$. The rest of the section is devoted to the identification of $(x_2(t))$. 
\begin{proposition}
For any continuous function $g$ on $[0,m^*]^2$, the relation
\[
\left(\int_0^t g(x,y)\mu_\infty(\diff x,\diff y)\,\diff u\right)\stackrel{\text{dist.}}{=}\left(\int_0^t g\left((m{-}x_2(u))(1{-}r,r)\right)\,\diff u\right)
\]
holds, with $r=\gamma/(\gamma+\gamma^*)$.
\end{proposition}
One concludes that the measure $\pi_u(\diff x,\diff y)$ of  Proposition~\ref{propmu} is simply the Dirac measure at $[m-x_2(u)](1-r,r)$. This is the rigorous description of the fact described at the beginning of this section that if the fraction of polymerized mass is $x_2(u)$ then the fraction of regular [resp. misfolded] monomers is $(1-r)(m-x_2(u))$ [resp. $r(m-x_2(u))$].
\begin{proof}
The criteria of the modulus of continuity shows  that the sequence of  processes
\[
\left( \int_0^t g\left(\overline{X}_0^{N_k}(u), \overline{X}_1^{N_k}(u)\right)\,\diff u \right)
\]
is tight. By convergence in distribution of $(\mu_{N_k})$, one has, for $t\geq 0$, 
\begin{multline}\label{eqaux5}
\lim_{k\to +\infty} \int_0^t g\left(\overline{X}_0^{N_k}(u), \overline{X}_1^{N_k}(u)\right)\,\diff u \\
\int_0^t g(x,y)\pi_u(\diff x,\diff y)\,\diff u =
\int_0^t g\left(x,\frac{\gamma}{\gamma^*}x\right)\pi_u(\diff x,\diff y)\,\diff u
\end{multline}
by Proposition~\ref{proppi}. The same convergence in distribution also holds for finite marginals. One has to identify the first marginal of  $(\pi_u)$. If $f$ is a continuous function on $[0,m^*]$, by conservation of mass, one has the relation
\[
\left(\int_0^t f\left(\overline{X}_0^{N_k}(u)+\overline{X}_1^{N_k}(u)\right)\,\diff u\right)= 
\left(\int_0^t f\left(\frac{M_{N_k}}{N_k} - \overline{X}_2^{N_k}(u)\right)\,\diff u\right). 
\]
Relation~\eqref{eqaux5} and the convergence properties of the right hand side of this identity give the following identity of processes
\[
\left(\int_0^t f\left(x/r \right)\pi_u(\diff x,\diff y)\,\diff u\right)= 
\left(\int_0^t f\left(m -x_2(u)\right)\,\diff u\right).
\]
The proposition is proved. 
\end{proof}

\begin{theorem}\label{ThMF}
Under the scaling condition~\eqref{scaling} and if  the initial state of the solution $(X^N(t))$ of the SDE~\eqref{SDEMF}  is  $X^N(0)=(M_N,0,0)$ then, for the convergence in distribution,
\begin{equation}\label{limMF}
\lim_{N\to+\infty} \left(\frac{X_2^N(Nt)}{N}\right)= (x_2(t))\stackrel{\text{def.}}{=} \left( \frac{1-e^{-\beta r m t }}{1 + (\beta/\alpha r-1) e^{-\beta r m t }}m\right),
\end{equation}
with $r=\gamma/(\gamma+\gamma^*)$. 
\end{theorem}
\begin{proof}
By using Relation~\eqref{SDEX22}, Proposition~\ref{propmu} and the above proposition, one gets that any limiting point $(x_2(t))$ of $({X_2^N(Nt)}/{N})$  satisfies necessarily the following integral equation (integral form of the equation~\eqref{eq:intermediaire})
\begin{equation}\label{eqaux7}
x_2(t)=\alpha r^2\int_0^t (m-x_2(s))^2\,\diff s+\beta r\int_0^t (m-x_2(s))x_2(s)\,\diff s.
\end{equation}
By uniqueness of the solution of this equation, one gets the convergence in distribution of the sequence of processes  $({X_2^N(Nt)}/{N})$. Its explicit expression is easily obtained. 
\end{proof}
The following corollary gives the asymptotics of the first instant when a fraction $\delta\in(0,1)$ of monomers has been polymerized. This is a key quantity that can be measured with experiments. 
\begin{corollary}\label{lagLLN} [Asymptotics of Lag Time]
Under the conditions of Theorem~\ref{ThMF}, if for $\delta\in(0,1)$,
\begin{equation}\label{Tdelta}
T^N(\delta)=\inf\{t\geq 0: X_2^N(t)/M_N\geq \delta \},
\end{equation}
then, for the convergence in distribution
\begin{equation}\label{tdelta}
\lim_{N\to+\infty} \frac{T^N(\delta)}{N} = t_\delta\stackrel{\text{def.}}{=}\frac{1}{rm\beta}\log\left(1+\frac{\delta \beta}{\alpha r(1-\delta)}\right).
\end{equation}
\end{corollary}

\subsection{Central Limit Theorem}
From Proposition~\ref{proppi}, it has been proved that if $f:[0,m^*]^2$ is a ${\cal C}^1$-function then, for the convergence in distribution
\[
\lim_{N\to+\infty} \left(\int_0^t \left(\gamma^*y{-}\gamma x\right)\left(\frac{\partial}{\partial x}f(x,y){-}\frac{\partial}{\partial y}f(x,y)\right)\, \mu_N(\diff x,\diff y,\diff s)\right)=(0),
\]
with the above notations,
The following proposition is an extension of this result. This is the key ingredient to prove the central limit result of this section. 
\begin{proposition}\label{propclt}
If $g:[0,m^*]^2\times\R_+$ is a ${\cal C}^1$-function then, for the convergence in distribution,
\[
\lim_{N\to+\infty} \left(\int_0^t \left(\gamma^*y{-}\gamma x\right)\left(\frac{\partial}{\partial x}g(x,y,u){-}\frac{\partial}{\partial y}g(x,y,u)\right)\,\sqrt{N}\mu_N(\diff x,\diff y, \diff u)\right)=(0).
\]
\end{proposition}
\begin{proof}
We follow the same lines as in the proof of Proposition~\ref{proppi}.
The analogue of Relation~\eqref{eqaux3} is
\begin{multline}\label{eqaux6}
\frac{1}{\sqrt{N}}\left(g\left( \overline{X}^N(t),t\right) -g\left( \overline{X}^N(0),0\right)\right)\\
=\sqrt{N}\int_0^t \left(\gamma^*y{-}\gamma x\right)\left(\frac{\partial}{\partial x}g(x,y,s){-}\frac{\partial}{\partial y}g(x,y,s)\right)\,\mu_N(\diff x,\diff y, \diff s)\\
- \frac{\alpha}{\sqrt{N}} \int_0^t \overline{X}^N_1(s)\left(\overline{X}^N_1(s)-1/N\right) \frac{\partial f}{\partial y}\left(\overline{X}^N(s)\right)\, \diff s\\
-\frac{\beta}{\sqrt{N}} \int_0^t \overline{X}^N_1(s)\overline{X}^N_2(s) \frac{\partial f}{\partial y}\left(\overline{X}^N(s)\right)\, \diff s\\
+ \frac{1}{\sqrt{N}} \int_0^t \frac{\partial f}{\partial z}\left(\overline{X}^N(s),s\right)\, \diff s +\frac{M_g^N(Nt)}{\sqrt{N}}+o(1/\sqrt{N}),
\end{multline}
It is not difficult to check  with the analogue of Relation~\eqref{eqaux2} for the previsible increasing process of the martingale $({M_g^N(Nt)}/{\sqrt{N}})$ that, for $t\geq 0$, 
\[
\lim_{N\to+\infty} \E\left(\croc{\frac{M_g^N(Nt)}{\sqrt{N}}}\right)=
\lim_{N\to+\infty} \frac{\E\left(\croc{M_g^N}(Nt)\right)}{N}=0.
\]
Consequently, by Doob's Inequality, the martingale of Relation~\eqref{eqaux6} vanishes when $N$ gets large. The desired convergence of the proposition is then easily derived. 
\end{proof}

\begin{theorem}[Central Limit Theorem]\label{CLTMF}
Under Condition~\eqref{scaling} and if $(x_2(t))$ is the function defined by Relation~\eqref{limMF} then, for the convergence in distribution,
\[
\lim_{N\to+\infty} \left(\frac{X_2^N(Nt)-Nx_2(t)}{\sqrt{N}}\right)= (U(t)),
\]
where $(U(t))$ is the solution of the stochastic differential equation
\begin{equation}\label{U}
\diff U(t) = \sqrt{\sigma(t)}\diff B(t)+h(t) U(t)\,\diff t,
\end{equation}
and  $(B(t))$ is a standard Brownian motion and
\[
\begin{cases}
\sigma(t)= 2\alpha r^2 (m{-}x_2(t))^2{+}\beta r(m{-}x_2(t))x_2(t) \\
h(t)=r(\beta-2\alpha r)(m-x_2(t))-\beta r x_2(t). 
\end{cases}
\]
\end{theorem}
The corresponding result of Eug\`ene et al.~\cite{EXRD} when there is no misfolding phenomenon shows that the functions $\sigma$ and $h$ are similar if $\alpha$ and $\beta$ are respectively replaced by $\alpha r^2$ and $\beta r$. 
\begin{proof}
Denote
\[
U^N(t)=\frac{X_2^N(Nt)-Nx_2(t)}{\sqrt{N}}=\sqrt{N}\left(\overline{X}_2^N(t))-x_2(t)\right).
\]
By combining Equation~\eqref{SDEX22},
\[
\overline{X}_2^N(t))=\alpha\int_0^t \overline{X}^N_1(s)^2\,\diff s\\+\beta\int_0^t \overline{X}_1^N(s) \overline{X}_2^N(s)\,\diff s+\frac{M_2^N(Nt)}{N}+O(1/N)
\]
and Relation~\eqref{eqaux7},
\begin{equation}\label{eqaux8}
x_2(t)=\alpha r^2\int_0^t (m-x_2(s))^2\,\diff s+\beta r\int_0^t (m-x_2(s))x_2(s)\,\diff s,
\end{equation}
one gets
\begin{multline*}
U^N(t)= \alpha\sqrt{N}\int_0^t \left(\overline{X}^N_1(s)^2-r^2(m-x_2(s))^2\right)\,\diff s\\+\beta\sqrt{N}\int_0^t \left(\overline{X}_1^N(s) \overline{X}_2^N(s)-r(m-x_2(s))x_2(s)\right)\,\diff s+\frac{M_2^N(Nt)}{\sqrt{N}}+O(1/\sqrt{N}).
\end{multline*}
Concerning the martingale term,  Relation~\eqref{X2croc} gives, for $t\geq 0$,
\[
\croc{\frac{M_2^N}{\sqrt{N}}}(Nt)=2\alpha\int_0^t \overline{X}_1^N(s)^2\,\diff s+\beta \int_0^t \overline{X}_1^N(s) \overline{X}_2^N(s)\,\diff s+O(1/N).
\]
With the same method as in the proof of Theorem~\ref{ThMF}, one gets the following convergence in distribution
\[
\lim_{N\to+\infty} \left(\croc{\frac{M_2^N}{\sqrt{N}}}(Nt)\right)
{=}\left(2\alpha r^2\int_0^t (m{-}x_2(s))^2\,\diff s{+}\beta r \int_0^t x_2(s)(m{-}x_2(s))\,\diff s\right)
\]
by Relation~\eqref{eqaux8}.

Note also that, for $s\geq 0$,
\begin{multline*}
\sqrt{N}\left(\overline{X}_1^N(s) \overline{X}_2^N(s)-r(m-x_2(s))x_2(s)\right)\\=
U^N(s)\overline{X}_1^N(s)+\sqrt{N}\left(\overline{X}_1^N(s)-r(m-x_2(s))\right)x_2(s)
\end{multline*}
and
\begin{equation}\label{eqaux10}
\sqrt{N}\left(\overline{X}^N_1(s)-r(m-x_2(s)\right)=-\frac{\sqrt{N}}{\gamma+\gamma^*}\left(\gamma \overline{X}^N_0(s)-\gamma^* \overline{X}^N_1(s)\right)
-rU^N(t).
\end{equation}
The above relation for $(U^N(t))$ can then be rewritten as 
\begin{multline}\label{eqaux9}
U^N(t)= \int_0^t U^N(s)\left( (\beta-\alpha r)\overline{X}^N_1(s)-\alpha r^2(m-x_2(s))-\beta r x_2(s)\right)\,\diff s\\
{-}\frac{1}{\gamma+\gamma^*} \int_0^t \left(\gamma^*y{-}\gamma x\right)\left[\alpha (y{+}r(m{-}x_2(s))){+}\beta x_2(s)\right]\,\sqrt{N}\mu_N(\diff x,\diff y,\diff s)\\
+\frac{M_2^N(Nt)}{\sqrt{N}}+O(1/\sqrt{N}).
\end{multline}
The convergence in distribution of the martingale, Proposition~\ref{propclt} and the criterion of the modulus of continuity give easily the tightness of the sequence $(U^N(t))$. Let $(U(t))$ be a limit of some subsequence $(U^{N_k}(t))$.   

A close look at Relation~\eqref{eqaux9} shows that the theorem will be proved, with standard arguments, if the following convergence in distribution is proved
\[
\lim_{k\to+\infty} \left(\int_0^t U^{N_k}(s)\overline{X}^{N_k}_1(s)\,\diff s\right) = \left(r \int_0^t U(s)(m-x_2(s))\,\diff s\right).
\]
For $k\geq 0$,
\begin{align*}
\int_0^t& U^{N_k}(s)\overline{X}^{N_k}_1(s)-rU(s)(m-x_2(s))\,\diff s\\
&{=}\int_0^t U^N(s)\left(\overline{X}^N_1(s){-}r(m{-}x_2(s))\right)\,\diff s{+}\int_0^t r(m{-}x_2(s))\left(U^N(s){-}U(s)\right)\,\diff s,
\end{align*}
the process associated to the last term of the second part of this  identity converges in distribution to $0$. By Relation~\eqref{eqaux10}, the first term  can be written as 
\begin{multline*}
 -\int_0^t \left(\overline{X}^{N_k}_2(s)-x_2(s)\right)\left(\frac{\sqrt{{N_k}}}{\gamma+\gamma^*}\left(\gamma^* \overline{X}^{N_k}_0(s)-\gamma \overline{X}^{N_k}_1(s)\right)+rU^{N_k}(t)\right)\,\diff s\\
= -r\int_0^t \left(\overline{X}^{N_k}_2(s)-x_2(s)\right)U^{N_k}(s)\,\diff s\\ -\frac{1}{\gamma+\gamma^*} \int_0^t \left(\frac{M_{N_k}}{{N_k}}-x-y-x_2(s)\right)\left(\gamma x-\gamma^* y\right)\,\sqrt{{N_k}}\mu_{N_k}(\diff x,\diff y,\diff s).
\end{multline*}
the first term of the right hand side converges in distribution to $0$ due to Theorem~\ref{ThMF} and the same property also holds for the second term by Proposition~\ref{propclt}. The theorem is proved. 
\end{proof}

As a consequence, one gets the following central limit theorem for the lag time.  The notations of Corollary~\ref{lagLLN} and Theorems~\ref{ThMF} and~\ref{CLTMF} are used. 
\begin{corollary} 
Under the scaling regime~\eqref{scaling},  for $\delta\in(0,1)$,
the convergence in distribution
\begin{equation}\label{eqVT}
\lim_{N\to+\infty} \frac{T^N(\delta)-N t_\delta}{\sqrt{N}} = 
\frac{\delta\eta-U(t_\delta)}{r m^2(1-\delta)(\beta+\alpha r(1-\delta))}
\end{equation}
holds, where the variables  $T^N(\delta)$ and $t_\delta$ are defined by~\eqref{Tdelta} and~\eqref{tdelta} and $(U(t))$ by~\eqref{U}, and $r=\gamma/(\gamma+\gamma^*)$. 
\end{corollary}
\begin{proof}
For $z\in\R$ note that, since $(X_2^N(t))$ is a non-decreasing  process, 
\begin{multline*}
\left\{\frac{T^N(\delta)-N t_\delta}{\sqrt{N}} \geq z\right\}=
\left\{X_2^N\left(s_N\right) <\delta M_N\right\}\\
=\left\{\frac{\overline{X}_2^N\left(s_N\right)-Nx_2(s_N/N)}{\sqrt{N}} <\frac{\delta M_N -Nx_2(s_N/N)}{\sqrt{N}}\right\},
\end{multline*}
with $s_N=N t_\delta+z\sqrt{N}$. From Theorem~\ref{CLTMF} one gets the convergence in distribution
\[
\lim_{N\to+\infty} \frac{\overline{X}_2^N\left(s_N\right)-Nx_2(s_N/N)}{\sqrt{N}}=U(t_\delta)
\]
and the expansion of $(x_2(t))$ at $t_\delta$ gives
\[
\lim_{N\to+\infty} \frac{\delta M_N -Nx_2(s_N/N)}{\sqrt{N}}=\delta\eta -z r m^2(1-\delta)(\beta+\alpha r(1-\delta)).
\]
This completes the proof of the corollary.
\end{proof}
Equation~\eqref{eqVT} shows that the variance of the lag time is inversely proportional to $\gamma/\gamma^*$, a low misfolding rate will thus increase the variability of the polymerisation process. 
\section{Models with Scaled Reaction Rates}\label{AlpSec}
For $t\geq 0$,  $X^N_1(t)$ is the number of  monomers at time $t$ and $X^N_2(t)$ is the number of polymerized monomers. The initial condition is $X_1^N(0)=M_N$ and $X_2^N(0)=0$. Because of the relation of conservation of mass, one has $M_N=X_1^N(t)+X_2^N(t)$. 

It is not difficult to see that the process $(X_2^N(t))$  can be represented as the solution of the following stochastic differential equations, 
\begin{equation}\label{SDEAlp}
\diff X_2^N(t) =  2 \sum_{i=1}^{{X_1^N(X_1^N{-}1)(s{-})}/{2}} \mathcal{N}_{{\alpha}/{N^{\nu+2}}}^i(\diff t){+}\sum_{i=1}^{X_1^N(s-)X_2^N(s{-}))} \mathcal{N}_{{\beta}/{N^2}}^i(\diff t).
\end{equation}
By integrating this equation, one gets the relation
\begin{equation}\label{eqa1}
X_2^N(t) = \frac{\alpha}{N^{2+\nu}}\int_0^t X_1^N(s)(X_1^N(s){-}1)\,\diff s +\frac{\beta}{N^2}\int_0^t X_1^N(s)X_2^N(s)\,\diff s+M^N(t),
\end{equation}
where $(M^N(t))$ is a martingale whose previsible increasing process is given by
\begin{equation}\label{eqa2}
\croc{M}_N(t)=2\frac{\alpha}{N^{2+\nu}}\int_0^t X_1^N(s)(X_1^N(s){-}1)\,\diff s +\frac{\beta}{N^2}\int_0^t X_1^N(s)X_2^N(s)\,\diff s.
\end{equation}
The following proposition shows that, on the time scale $t\mapsto Nt$, the polymerised mass is for this model in the order of $N^{1-\nu}$. 
\begin{proposition}
Under the scaling condition~\eqref{scaling}, for the convergence in distribution, the relation
\[
\lim_{N\to+\infty} \left(\frac{X_2^N(Nt)}{N^{1-\nu}}\right)=\left(\frac{\alpha m}{\beta}\left(e^{\beta m t} -1\right)\right)
\]
holds.
\end{proposition}
\begin{proof}
The proof is standard by using the identities~\eqref{eqa1} and~\eqref{eqa2}, and the relation $X_1^N(t){+}X_2^N(t){=}M_N$. See Eug\`ene et al.~\cite{EXRD} for example. 
\end{proof}
The following lemma introduces a branching process which will be helpful to estimate the order of magnitude in $N$  of the lag time 
\[
T^N(\delta)=\inf\{t\geq 0: X_2^N(t)/M_N\geq \delta \},
\]
for $0<\delta<1$. 
\begin{lemma}
For $a$, $b>0$, let $(W_{a,b}^N(t))$ be a pure birth process with birth rate 
\[ \frac{a}{N^\nu}+  \frac{b}{N}x\] in state $x\in\N$, with $W(0)=0$ and $0<\nu\leq 1$.
 If
\[
\tau_{a,b}^N(\delta)\stackrel{\text{def.}}{=}\inf\left\{t>0: W_{a,b}^N(t)\geq \delta N\right\},
\]
then the sequence $({\tau_{a,b}^N(\delta)}/{(N\log N)})$ converges in distribution to $\nu/b$. 
\end{lemma}
As it can be seen $(W_{a,b}^N(t))$ is a branching process with immigration. Immigration rate is $a/N^{\nu}$ and the reproduction rate is given by $b/N$.  See Harris~\cite{Harris} for example. 
\begin{proof}
Let, for $x\in \N$, $E_x^N$ denotes an exponential random variable with parameter $a/N^\nu+x b/N$, assuming that the random variables $E_x^N$, $x\geq 0$ are independent, then clearly
\[
\tau_{a,b}^N(\delta)\stackrel{\text{dist}}{=} \sum_{x=0}^{\lfloor \delta N\rfloor} E_x^N.
\]
hence after some simple estimations
\[
\lim_{N\to+\infty} \frac{\E(\tau_{a,b}^N(\delta))}{N\log N}= \frac{\nu}{b}.
\]
In the same way, one checks that the sequence $(\Var(\tau_{a,b}^N(\delta)/N))$ is bounded 
\begin{equation}\label{var}
\Var \left(\frac{\tau_{a,b}^N(\delta)}{N} \right) \leq  \sum_{x=0}^{+\infty} \frac{1}{(a N^{1-\nu}+x b)^2}.
\end{equation}
The convergence in distribution follows , by using Chebishev's Inequality. 
\end{proof}
\begin{figure}[ht]
\centering
\includegraphics[scale=0.9]{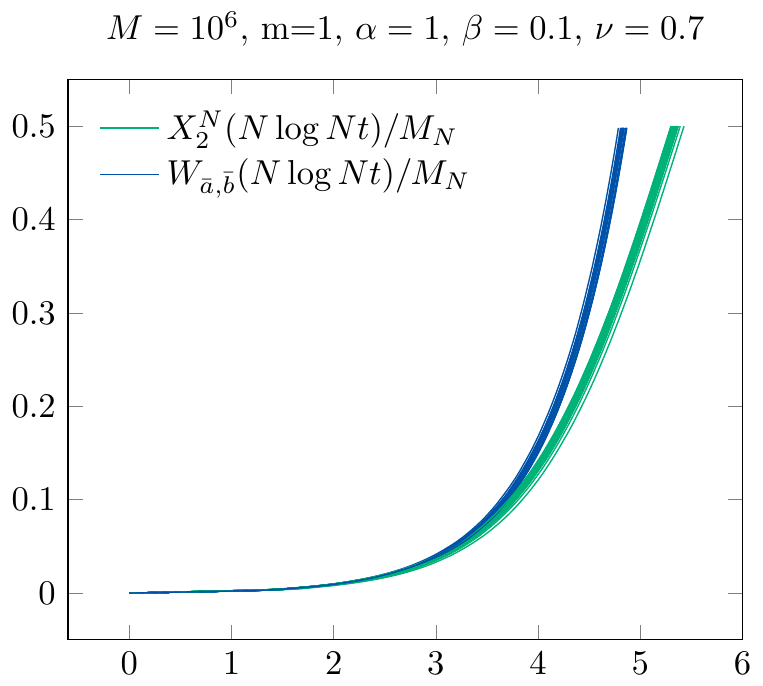}
\caption{In blue, $20$ simulations of $(W_{\bar{a},\bar{b}}/M_N)$ and in green, $20$ simulations of $(X_2^N/M_N)$  on the time scale $t\mapsto N \log N t$.}\label{Fig}
\end{figure}

Let $0<\delta<1$ and fix some $\underline{\kappa}<1<\overline{\kappa}$, one can assume that $N$ is sufficiently large so that  $\underline{\kappa}{\leq}M_N/(mN){\leq }\overline{\kappa}$ holds. 
Recall that $T^N(\delta)$ is the first time that the fraction of the number of polymerised monomers $X_2^N(t)/M_N$ is greater than $\delta$. 
The transition rates of $(X_2^N(t))$ are given by 
\begin{equation}\label{eqaux11}
 x \mapsto 
\begin{cases}
x{+}2\text{ at rate } \quad  \alpha/N^\nu\,[(M_N-x)/N]^2\\
x{+}1\phantom{ at rate } \quad  \beta\,x /N\times (M_N-x)/N.
\end{cases}
\end{equation}
By comparing the transition rates, we see that, for $x<\delta M_N$ one has 
\[
\begin{cases}
\alpha/N^\nu\,((M_N-x)/N)^2\ge \alpha/N^\nu\, (\underline{\kappa} m)^2(1-\delta)^2,\\ 
\beta\,x /N\times (M_N-x)/N \ge \beta \underline{\kappa} m (1-\delta). 
\end{cases}
\]
One can therefore construct a coupling such that, on the event $\{T^N(\delta)>t\}$, the relation $X_2^N(t){\geq} W_{\underline{a},\underline{b}}(t)$ holds with $\underline{a}{=}\alpha \underline{\kappa} m^2(1{-}\delta)^2$ and $\underline{b}{=}\beta \underline{\kappa} m(1{-}\delta)$. One obtains the relation $\tau_{\underline{a},\underline{b}}^N(\delta m){\geq}_{st} T^N(\delta)$, where $\geq_{st}$ denotes the stochastic order: if $U$ and $V$ are two real valued random variables
\[
U\geq_{st} V \text{ if } \P(V\geq x)\leq \P(U\geq x) \qquad \forall x\in\R.
\]
Since $X_2^N(t){\leq} 2 W_{\bar{a},\bar{b}}(t)$, with $\bar{a}{=}\alpha(\overline{\kappa}m)^2$ and $\bar{b}{=}\beta\overline{\kappa} m $, one has $\tau_{\bar{a},\bar{b}}^N(\delta m /2){\leq}_{st} T^N(\delta)$. One gets therefore
\begin{equation}\label{EQ}
\displaystyle\tau_{\bar{a},\bar{b}}^N(\delta m /2)\leq_{st} T^N(\delta) \leq_{st} \tau_{\underline{a},\underline{b}}^N(\delta m).
\end{equation}
Since the constants $\underline{\kappa}$ and $\overline{\kappa}$ can be chosen arbitrarily close to $1$, the following proposition has therefore been proved. 
\begin{proposition}[Order of Magnitude of Lag Time]
For $\delta>0$ and $ 0{<}\nu{\leq}1$, 
\[
\lim_{N\to+\infty}
\P\left(\frac{\nu}{\beta m }\leq \frac{T^N(\delta)}{N\log N}\leq \frac{\nu}{\beta m (1-\delta)}\right)=1.
\]
\end{proposition}
\noindent 
{\bf Remark}. It is very likely that, to reach the state $\delta N$, only the second reaction has a real impact as soon as the variable $X_2^N$ is not $0$.  If true, simple calculations, as in the proof of the above lemma, would then give that the variable $T^N(\delta)/(N\log N)$ is converging in distribution to $\nu/(\beta m)$ as $N$ get large. Note that the limit in this asymptotic result does not depend on $\delta$ which suggests a sharp transition for the polymerisation process. 

The birth process $(W_{\bar{a},\bar{b}}(t))$ seems to be close to $(X_2^N(t))$  during the initiation of the polymerisation, as  the simulations of Figure~\ref{Fig}. 
This suggests that, for $\delta$ small,  the variables $\tau_{\bar{a},\bar{b}}^N(\delta m)$ and $T^N(\delta)$ are very close. We conclude this part by considering the case $\nu>1$. 
\subsection*{A Very Slow Nucleation Step}
Now we assume that $\nu>1$, in this regime, the first reaction, the nucleation step, is then significantly  slowed. 
\begin{proposition}
For any $\eps>0$ and $0<\delta<1$, there exist $0<K_1<K_2$ such that 
\[
\liminf_{N\to+\infty}
\P\left(K_1\leq \frac{T^N(\delta)}{N^\nu}\leq K_2\right)\geq 1-\eps.
\]
\end{proposition}
\begin{proof}
By using Relation~\eqref{EQ}, it is enough to derive a corresponding limit theorem for $\tau_{a,b}^N(\delta)/{N^{\nu}}$ for some $a{>}0$ and $b{>}0$. Let $(E^1_x)$ be a sequence of i.i.d. exponential random variables with parameter $1$, then 
\begin{equation}\label{bignu}
\frac{\tau_{a,b}^N(\delta)}{N^{\nu}}=  \sum_{x=0}^{\lfloor \delta N\rfloor} \frac{E^1_x}{a +x b N^{\nu-1}} = \frac{E^1_0}{a} +  \sum_{x=1}^{\lfloor \delta N\rfloor} \frac{E^1_x}{a +x b N^{\nu-1}}.
\end{equation}
The expected value of the last term of the right hand side of the above relation is bounded by $K\log(N)/N^{\nu-1}$ for some constant $K{>}0$. Consequently, this term becomes negligible in distribution for $N$ large.  One gets that the variable ${\tau_{a,b}^N(\delta)}/{N^{\nu}}$ converges in distribution to an exponential random variable. The proposition is proved. 
\end{proof}
As we have seen in the proof, the only term that matters in the series in Relation~\eqref{bignu} is the first one: the time to reach one polymerised monomer. It characterises  the order of magnitude of the lag time.  This variable has been analysed in Szavits-Nossan et al.~\cite{Szavits} and Yvinec et al.~\cite{yvinec2016}.

\medskip
{\bf Acknowledgments.} We thank W.F. Xue (University of Kent) for inspiring discussions. M. Doumic and S. Eug\`ene's research was supported by ERC Starting Grant SKIPPER$^{AD}$ No. 306321.

\providecommand{\bysame}{\leavevmode\hbox to3em{\hrulefill}\thinspace}
\providecommand{\MR}{\relax\ifhmode\unskip\space\fi MR }
\providecommand{\MRhref}[2]{%
  \href{http://www.ams.org/mathscinet-getitem?mr=#1}{#2}
}
\providecommand{\href}[2]{#2}


\begin{thebibliography}{10}

\bibitem{AndersonKurtz}
David~F. Anderson and Thomas~G. Kurtz, \emph{Continuous time {M}arkov chain
  models for chemical reaction networks}, Design and Analysis of Biomolecular
  Circuits (Heinz Koeppl, Gianluca Setti, Mario di~Bernardo, and Douglas
  Densmore, eds.), Springer New York, 2011, pp.~3--42.

\bibitem{Billingsley}
P.~Billingsley, \emph{Convergence of probability measures}, second ed., Wiley
  Series in Probability and Statistics: Probability and Statistics, John Wiley
  \& Sons Inc., New York, 1999, A Wiley-Interscience Publication.

\bibitem{Bingham}
N.~H. Bingham, \emph{Fluctuation theory for the {E}hrenfest urn}, Advances in
  Applied Probability \textbf{23} (1991), no.~3, 598--611.

\bibitem{Boza}
Perinur Bozaykut, Nesrin~Kartal Ozer, and Betul Karademir, \emph{Regulation of
  protein turnover by heat shock proteins}, Free Radic Biol Med. \textbf{77}
  (2014), 195--209.

\bibitem{Dawson}
Donald~A. Dawson, \emph{Measure-valued {M}arkov processes}, \'Ecole d'\'Et\'e
  de Probabilit\'es de Saint-Flour XXI---1991, Lecture Notes in Math., vol.
  1541, Springer, Berlin, 1993, pp.~1--260.

\bibitem{DM}
Claude Dellacherie and Paul-Andr{\'e} Meyer, \emph{Probabilities and
  potential}, North-Holland Mathematics Studies, vol.~29, North-Holland
  Publishing Co., Amsterdam-New York; North-Holland Publishing Co.,
  Amsterdam-New York, 1978.

\bibitem{Dobson}
Christopher~M. Dobson, \emph{Protein folding and misfolding}, Nature
  \textbf{426} (2003), 884--890.

\bibitem{Dobson:2}
ChristopherM. Dobson, \emph{The generic nature of protein folding and
  misfolding}, Protein Misfolding, Aggregation, and Conformational Diseases
  (VladimirN. Uversky and AnthonyL. Fink, eds.), Protein Reviews, vol.~4,
  Springer US, 2006, pp.~21--41 (English).

\bibitem{Eden2015}
Kym Eden, Ryan Morris, Jay Gillam, Cait~E. MacPhee, and Rosalind~J. Allen,
  \emph{Competition between primary nucleation and autocatalysis in amyloid
  fibril self-assembly}, Biophysical Journal \textbf{108} (2015), no.~3, 632 --
  643.

\bibitem{EXRD}
Sarah Eug\`ene, Wei-Feng Xue, Philippe Robert, and Marie Doumic, \emph{Insights
  into the variability of nucleated amyloid polymerization by a minimalistic
  model of stochastic protein assembly}, Submitted to Journal of Chemical
  Physics, September 2015.

\bibitem{Freidlin}
M.~I. Freidlin and A.~D. Wentzell, \emph{Random perturbations of dynamical
  systems}, second ed., Springer-Verlag, New York, 1998, Translated from the
  1979 Russian original by Joseph Sz\"ucs.

\bibitem{Harris}
Theodore~E. Harris, \emph{The theory of branching processes}, Dover Phoenix
  Editions, Dover Publications, Inc., Mineola, NY, 2002, Corrected reprint of
  the 1963 original.

\bibitem{Higham}
Desmond~J. Higham, \emph{Modeling and simulating chemical reactions}, SIAM
  Review \textbf{50} (2008), no.~2, 347--368.

\bibitem{Hunt}
P.J. Hunt and T.G Kurtz, \emph{Large loss networks}, Stochastic Processes and
  their Applications \textbf{53} (1994), 363--378.

\bibitem{KKP}
Hye-Won Kang, Thomas~G. Kurtz, and Lea Popovic, \emph{Central limit theorems
  and diffusion approximations for multiscale {M}arkov chain models}, The
  Annals of Applied Probability \textbf{24} (2014), no.~2, 721--759.

\bibitem{Karlin}
Samuel Karlin and James McGregor, \emph{Ehrenfest urn models}, Journal of
  Applied Probability \textbf{2} (1965), 352--376.

\bibitem{Knowles:2}
Tuomas P.~J. Knowles, Michele Vendruscolo, and Christopher~M. Dobson, \emph{The
  amyloid state and its association with protein misfolding diseases}, Nature
  Reviews Molecular Cell Biology \textbf{15} (2014), 384--396.

\bibitem{Kurtz}
T.G. Kurtz, \emph{Averaging for martingale problems and stochastic
  approximation}, Applied Stochastic Analysis, US-French Workshop, Lecture
  notes in Control and Information sciences, vol. 177, Springer Verlag, 1992,
  pp.~186--209.

\bibitem{Lanneau}
David Lanneau, Guillaume Wettstein, Philippe Bonniaud, and Carmen Garrido,
  \emph{Heat shock proteins: cell protection through protein triage},
  ScientificWorldJournal \textbf{10} (2010), 1543--1552.

\bibitem{McManus_2016}
Jennifer~J. McManus, Patrick Charbonneau, Emanuela Zaccarelli, and Neer
  Asherie, \emph{The physics of protein self-assembly}, preprint, February
  2016.

\bibitem{Morando}
Philippe Morando, \emph{Mesures al\'eatoires}, S\'eminaire de Probabilit\'es de
  Strasbourg \textbf{III} (1969), 190--229.

\bibitem{Ow_Protein2014}
Sian-Yang Ow and Dave~E. Dunstan, \emph{A brief overview of amyloids and
  alzheimer's disease}, Protein Science \textbf{23} (2014), no.~10, 1315--1331.

\bibitem{PSV}
G.~C. Papanicolaou, D.~Stroock, and S.~R.~S. Varadhan, \emph{Martingale
  approach to some limit theorems}, Papers from the {D}uke {T}urbulence
  {C}onference ({D}uke {U}niv., {D}urham, {N}.{C}., 1976), {P}aper {N}o. 6,
  Duke Univ., Durham, N.C., 1977, pp.~ii+120 pp. Duke Univ. Math. Ser., Vol.
  III.

\bibitem{Pigolotti2013}
Simone Pigolotti, Ludvig Lizana, Daniel Otzen, and Kim Sneppen, \emph{Quality
  control system response to stochastic growth of amyloid fibrils}, \{FEBS\}
  Letters \textbf{587} (2013), no.~9, 1405 -- 1410.

\bibitem{Rogers}
L.~C.~G. Rogers and David Williams, \emph{Diffusions, {M}arkov processes, and
  martingales. {V}ol. 1: Foundations}, second ed., John Wiley \& Sons Ltd.,
  Chichester, 1994.

\bibitem{Rudin}
Walter Rudin, \emph{Real and complex analysis}, third ed., McGraw-Hill Book
  Co., New York, 1987.

\bibitem{Sun}
Wen Sun, Mathieu Feuillet, and Philippe Robert, \emph{Analysis of large
  unreliable stochastic networks}, Annals of Applied Probability (2015), To
  Appear.

\bibitem{Szavits}
Juraj Szavits-Nossan, Kym Eden, Ryan~J. Morris, Cait~E. MacPhee, Martin~R.
  Evans, and Rosalind~J. Allen, \emph{Inherent variability in the kinetics of
  autocatalytic protein self-assembly}, Physical Review Letters \textbf{113}
  (2014), 098101.

\bibitem{Radford}
W-F Xue, S~W Homans, and S~E Radford, \emph{Systematic analysis of
  nucleation-dependent polymerization reveals new insights into the mechanism
  of amyloid self-assembly}, PNAS \textbf{105} (2008), 8926--8931.

\bibitem{yvinec2016}
Romain Yvinec, Samuel Bernard, Erwan Hingant, and Laurent Pujo-Menjouet,
  \emph{First passage times in homogeneous nucleation: Dependence on the total
  number of particles}, The Journal of Chemical Physics \textbf{144} (2016),
  no.~3, 034106.

\end{thebibliography}
\end{document}